\documentclass[11pt,fleqn]{article}

\usepackage{geometry}
\usepackage{graphicx}
\usepackage{verbatim}
\usepackage[utf8]{inputenc}

\usepackage{amsmath,amsfonts,amsthm,latexsym,amssymb,enumerate,float}
\usepackage{subcaption}

\usepackage[colorlinks = true, citecolor = blue, linkcolor = black, breaklinks = true]{hyperref,xcolor}
\usepackage[overload]{empheq}
\usepackage{breqn}
\usepackage{enumitem}
\usepackage{complexity}

\usepackage{rotating}
\usepackage{marvosym}

\usepackage{tikz}
\usetikzlibrary{arrows.meta, automata, shapes, positioning}
\usepackage{txfonts}

\newtheorem{theorem}{Theorem}
\newtheorem{corollary}[theorem]{Corollary}
\newtheorem{proposition}[theorem]{Proposition}

\newtheorem{definition}[theorem]{Definition}

\usepackage{abstract}
\renewcommand{\abstractname}{}    

\renewenvironment{abstract}
{\quotation\small\noindent\rule{\linewidth}{.5pt}\par
	{\centering\bfseries\abstractname\par}
}
{\par\noindent\rule{\linewidth}{.5pt}\endquotation}

\newcommand{\probld}[3]{
	\begin{flushleft}
		\fbox{
			\begin{minipage}{.96\textwidth}
				\noindent {\sc #1}\\
				{\bf Input:} #2\\
				{\bf Question:} #3
			\end{minipage}
		}
	\end{flushleft}
}

\begin{document}

\title{On the Computational Complexity of the Strong Geodetic Recognition Problem\thanks{This study was financed in part by the Coordena\c{c}\~ao de Aperfei\c{c}oamento de Pessoal de N\'ivel Superior - Brasil (CAPES) - Finance Code 001,
					Conselho Nacional de Desenvolvimento Científico e Tecnológico (CNPq), and Fundação de Amparo à Pesquisa do Estado de Minas Gerais (FAPEMIG).}}
\author{Carlos {V.G.C. Lima}$^1$ \hspace{.2cm}
	Vinicius {F. dos Santos}$^2$ \hspace{.2cm}
	João {H.G. Sousa}$^2$ \hspace{.2cm}
	Sebastián {A. Urrutia}$^2$}
\date{}\maketitle\vspace{-1cm}
\begin{center}
	{\small
		$^1$ Centro de Ciências e Tecnologia, Universidade Federal do Cariri, Juazeiro do Norte, Brazil\\ vinicius.lima@ufca.edu.br\\
		$^2$ Departamento de Ciência da Computação, Universidade Federal de Minas Gerais, Belo Horizonte, Brazil\\ $\{$viniciussantos, surrutia$\}$@dcc.ufmg.br, \, joao.gsousa77@gmail.com
	}
\end{center}

\begin{abstract}
	\noindent\textbf{Abstract:} A strong geodetic set of a graph~$G=(V,E)$ is a vertex set~$S \subseteq V(G)$ in which it is possible to cover all the remaining vertices of~$V(G) \setminus S$ by assigning a unique shortest path between each vertex pair of~$S$.
	In the \textsc{Strong Geodetic} problem (SG) a graph~$G$ and a positive integer~$k$ are given as input and one has to decide whether~$G$ has a strong geodetic set of cardinality at most~$k$.
	This problem is known to be \NP-hard for general graphs.
	In this work we introduce the \textsc{Strong Geodetic Recognition} problem (SGR), which consists in determining whether even a given vertex set~$S \subseteq V(G)$ is strong geodetic.
	We demonstrate that this version is \NP-complete.
	We investigate and compare the computational complexity of both decision problems restricted to some graph classes, deriving polynomial-time algorithms, \NP-completeness proofs, and initial parameterized complexity results, including an answer to an open question in the literature for the complexity of SG for chordal graphs.	
\end{abstract}
{\small
	\begin{tabular}{lp{12.5cm}}
		\textbf{Keywords:} & Covering $\cdot$ \NP-completeness $\cdot$ Strong geodetic number $\cdot$ \textsc{Strong Geodetic Recognition}
	\end{tabular}
}

\section{Introduction}
\label{intro}

    Determining efficient ways to cover vertices or edges of a graph gives rise to an important class of graph problems, that includes the classical vertex cover problem, one of Karp’s~21 \NP-complete problems~\cite{VCSurvey, Karp1972, Fellows2018}.
    Other examples of such problems include the covering of the vertices by independent sets (vertex coloring problem) or cliques (clique cover problem).
    
    Denoting by~$V(H)$ the vertex set of a subgraph~$H$ of~$G$, in this paper we consider problems of finding a covering of the vertices of a graph~$G$ by paths, that is, a family~$\mathcal{P}$ of distinct nonempty paths~$P_i$ (not necessarily vertex-disjoint), such that~$\bigcup V(P_i) = V(G)$.
    When the paths of~$\mathcal{P}$ are required to be shortest paths (geodesics), then it is known as the \textit{isometric path cover} problem~\cite{manuel2018strong}.
    
    In~\cite{harary1993geodetic}, the authors introduced the \textsc{Geodetic} problem, which consists of determining the minimum cardinality of a vertex set~$S \subseteq V(G)$ of a graph~$G=(V,E)$, such that every vertex of~$G$ lies on a shortest path between a pair of vertices in~$S$.
    Such a set~$S$ is called a \textit{geodetic set} of~$G$.
    The minimum cardinality~$g(G)$ of a geodetic set is the \textit{geodetic number} of~$G$.
    A \textit{$g$-set} is a geodetic set of minimum cardinality.
    
    In~\cite{atici2002computational} the \textsc{Geodetic} problem was proved to be \NP-hard even for graphs of diameter~2.
    More recently, in~\cite{dourado2010some}, the authors proved that the problem is also \NP-hard when restricted to chordal graphs and to chordal bipartite graphs.
    Moreover, in the same paper the authors achieved exact values concerning the geodetic number of split graphs and give a linear time algorithm for cographs, besides some upper bounds, particularly for unit interval graphs.
    This upper bound has been improved by Ekim et al.~\cite{Ekim2012}, that provide a polynomial-time algorithm for the \textsc{Geodetic} problem.
    In~\cite{ekim2014block}, a polynomial-time algorithm for the \textsc{Geodetic} problem restricted to block-cacti graphs is presented, besides an \NP-hardness proof for cobipartite graphs.
    In~\cite{BUENO201822} the authors show the \NP-hardness for subcubic graphs, that is, graphs of maximum degree~3.
    In~\cite{hernando2005steiner} a comparison between the hull, Steiner, and geodetic numbers of graphs is given.
    Bre\v{s}ar et al.~\cite{BRESAR20085555} determined some exact values and upper bounds for the geodetic number of the Cartesian product of graphs.
    Cao et al.~\cite{Cao2009} presented exact values for the geodetic number of the Cartesian product of cycles.
    
    Other variations of the problem have been proposed in the literature, where we can cite the edge version~\cite{Atici2010}, for oriented graphs~\cite{HUNG20092134}, and the connected geodetic number~\cite{SANTHAKUMARAN20091571}.
	In this paper we study another variation
    defined by Manuel et al.~\cite{manuel2018strong}, where a unique shortest path between each vertex pair of~$S \subseteq V(G)$ is assigned to cover the vertices of a graph~$G$.
    
    More formally, let~$G=(V,E)$ be a simple undirected graph.
    For~$u, v \in V$, we denote~$P(u,v)$ as the set containing all shortest paths between~$u$ and~$v$ in~$G$.
    For a set~$S \subseteq V$, let~$U_S$ be the set of all distinct vertex pairs of~$S$.
    We say that~$I(S)$ is a \textit{shortest path assignment of~$S$} if
	    	\begin{dmath} \label{eq:IS}
	    		I(S) = \Bigl\{P_1, P_2, \dotsc, P_{|U_S|} \, { \mid } \, {\Big(P_i \in P(u_i,v_i) \wedge P_j \in P(u_j,v_j) \Longleftrightarrow (u_i \neq u_j) \vee (v_i \neq v_j) \Big),}\\
	    		{\textnormal{for all } (u_i,v_i), (u_j, v_j) \in U_S \textnormal{ and } i \neq j \Bigr\},}
	    	\end{dmath}
    that is, $I(S)$ is a shortest path assignment for~$S$ if it contains a unique shortest path between~$u$ and~$v$, for each pair of distinct vertices~$(u,v)$ of~$S$.

     A vertex set~$S$ is a \textit{strong geodetic set} of~$G$ if there exists a shortest path assignment~$I(S)$ of~$S$, such that
        $\bigcup_{P \in I(S)}{V(P)} = V(G)$.
    An \textit{sg-set} is a strong geodetic set of minimum cardinality.
    We denote by~$sg(G)$ the cardinality of an \textit{sg}-set of a given graph~$G$, that is, the \textit{strong geodetic number} of~$G$.
    
    As observed by Manuel et al.~\cite{manuel2018strong}, since every strong geodetic set is geodetic as well, it follows that~$sg(G) \geq g(G)$, for every graph~$G$.
    Moreover, the equality holds for \textit{geodetic graphs}, that is, graphs where the shortest path between any two vertices is unique.
    The family of geodetic graphs~\cite{Blokhuis88,nebesky2002new,Plesnik84} includes, for example, block graphs~\cite{BEHTOEI2010219}, that are equivalent to diamond-free chordal graphs.
    Then, it follows by~\cite{ekim2014block} that~$sg(G)$ can be determined in polynomial time for block graphs.
    However, the gap between the two parameters can be arbitrarily large, as depicted in Figure~\ref{fig:K_2n}, which shows that~$g(K_{2,n} = 2$ and~$sg(K_{2,n} = n$, for every~$n \geq 3$.
    In order to prove this, we first emphasize that it is not hard to see that at least two vertices must be in any $s$-set or $sg$-set for graphs with at least two vertices.
    In the complete bipartite graph~$G=K_{2,n}$, $n\geq 3$, it follows that the~$g$-set of~$G$ is unique and formed by the two vertices of the smallest part of its bipartition (vertices~$u$ and~$w$ in Figure~\ref{fig:gSetK_2}), which implies that~$g(G) = 2$.
    On the other hand, every vertex~$v_i$ in Figure~\ref{fig:sgSetK_2} must be included in any strong geodetic set~$S$ of~$G$, otherwise~$u$ and~$w$ must be in~$S$, and then the $u,w$-geodesic in~$I(S)$ must contain~$v_i$.
    In this way, each~$v_j$, $j \neq i$, must be in~$S$.
    We can obtain a strong geodetic set~$S'$ from~$S$ by removing~$u$ and~$w$ and adding~$v_i$ (see Figure~\ref{fig:sgSetK_2}).
    It is not hard to see that~$S'$ is an~$sg$-set of~$G$ (see Corollary 2.3 of~\cite{multipartite}).
    
    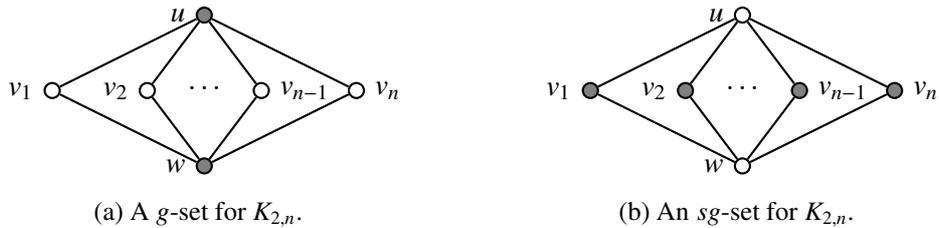
\begin{figure}[b!]
	\centering
	\begin{subfigure}[b]{0.35\textwidth}
	    \centering
    	\begin{tikzpicture}
    	
    	    \tikzset{node/.style={draw, fill=white, circle, thick, fill = white, auto=left,
            inner sep=2pt}}
            
    	    \node[node] (u) at (0,0) [fill=gray, label=left:$u$] {};
    	    \node[node] (v1) at (-2,-1) [label=left:$v_1$] {};
    	    \node[node] (v2) at (-0.75,-1) [label=left:$v_2$] {};
		    \node[circle, thick, fill = white, auto=left, inner sep=2pt] (ret) at (0,-1) {$\cdots$};
	    	\node[node] (v4) at (0.75,-1) [label=right:$v_{n-1}$] {};
    		\node[node] (v5) at (2,-1) [label=right:$v_n$]{};
    	    \node[node] (w) at (0,-2) [fill=gray, label=left:$w$] {};
    	    
    	    \foreach \source/\dest in {u/v1, u/v2, u/v4, u/v5,
                                        v1/w, v2/w, v4/w, v5/w}
                \path (\source) edge[thick] node {} (\dest);
		    
	    \end{tikzpicture}
	    \caption{A $g$-set for $K_{2,n}$.}
	    \label{fig:gSetK_2}
	\end{subfigure} \qquad \qquad
	~
	\begin{subfigure}[b]{0.35\textwidth}
	    \centering
		\begin{tikzpicture}[level/.style={sibling distance=20mm/#1}]
		
		    \tikzset{node/.style={draw, fill=white, circle, thick, fill = white, auto=left,
            inner sep=2pt}}
            
    	    \node[node] (u) at (0,0) [label=left:$u$] {};
    	    \node[node] (v1) at (-2,-1) [fill=gray, label=left:$v_1$] {};
    	    \node[node] (v2) at (-0.75,-1) [fill=gray, label=left:$v_2$] {};
		    \node[circle, thick, fill = white, auto=left, inner sep=2pt] (ret) at (0,-1) {$\cdots$};
	    	\node[node] (v4) at (0.75,-1) [fill=gray, label=right:$v_{n-1}$] {};
    		\node[node] (v5) at (2,-1) [fill=gray, label=right:$v_n$]{};
    	    \node[node] (w) at (0,-2) [label=left:$w$] {};
    	    
    	    \foreach \source/\dest in {u/v1, u/v2, u/v4, u/v5,
                                        v1/w, v2/w, v4/w, v5/w}
                \path (\source) edge[thick] node {} (\dest);
	        
        \end{tikzpicture}
		\caption{An $sg$-set for $K_{2,n}$.}
		\label{fig:sgSetK_2}
	\end{subfigure}
	\caption{In Figure~\ref{fig:gSetK_2} and Figure~\ref{fig:sgSetK_2} the gray vertices represent the $g$-set and the $sg$-set of~$K_{2,n}$, respectively.}
	\label{fig:K_2n}
\end{figure}
    
    We consider the corresponding decision problem in this work.

    \vspace{-.1cm}
    \probld{Strong Geodetic (SG)}
    {A finite, simple, and undirected graph~$G$ and a positive integer $k$.}
    {Is there a strong geodetic set~$S$ of~$G$ with~$|S|~\leq k$?}
	\vspace{-.1cm}
	
	Manuel et al.~\cite{manuel2018strong} proved that SG is \NP-hard for general graphs and derived a closed formula for~$sg(G)$ on Apollonian networks~$G$.
    Afterwards, the strong geodetic number was studied for grid-like architectures and Cartesian product graphs~\cite{gledel2018strong,refId0,Klavzar2018}, where are given closed formulas to solve the problem for some restricted graph classes.
    In~\cite{Wang2019} is given an upper bound in terms of the connectivity.
    Balanced complete bipartite graphs was studied in~\cite{irvsivc2018strong} and a lower bound in terms of the diameter is given.
    The authors achieved a quadratic algorithm solving~SG for complete bipartite graphs, derived results for complete multipartite graphs, and proved the \NP-hardness of SG for (general) bipartite graphs and complete multipartite graphs~\cite{multipartite}.
    In~\cite{gledel2018strong2}, the strong geodetic number (the exact value) for complete bipartite graphs and crown graphs was determined, related results for hypercubes were also presented.
    Furthermore, the concept of \textit{strong geodetic cores} has been introduced in~\cite{gledel2018strong} and stronger results concerning the Cartesian product of graphs were derived.
    
    As stated before, in~\cite{manuel2018strong} the authors claim that~SG is in fact \NP-complete for general graphs, but actually they do not prove that~SG $\in$ \NP, so their result implies only its hardness.
    Note that a possible certificate for a \textbf{YES} instance~$(G, k)$ of~SG could be given by a family~$\mathcal{P}$ of~$\binom{k}{2}$ paths of~$G$.
    We can verify that~$\mathcal{P}$ is in fact a valid certificate for~SG by showing that each~$P \in \mathcal{P}$ is a geodesic between its endvertices, $\bigcup_{P \in \mathcal{P}}{V(P)} = V(G)$, the ordered pairs defined by the set of endvertices of the paths of~$\mathcal{P}$ are all distinct, and the union of such pairs is a vertex set~$S$ of cardinality at most~$k$.
    The set~$S$ is then a strongly geodetic set of~$G$ of size at most~$k$.
    Obviously, verifying this certificate can be done in polynomial time on the size of~$(G, k)$, which implies that~SG$ \in$ \NP~and, by the reduction of Manuel et al.~\cite{manuel2018strong}, it is \NP-complete for general graphs.
    
    On the other side, a certificate for a \textbf{YES} instance~$(G, k)$ of the decision version of the {\sc Geodetic} problem can be given by just a set~$S \subseteq V(G)$, where we can easily verify in polynomial time whether~$S$ is a geodetic set of~$G$ of size at most~$k$.
    So, a natural question arises, that is, what is the complexity of deciding whether a given vertex set is a solution for~SG?
    Hence, we introduce the following decision problem.
    \vspace{-.1cm}
    \probld{Strong Geodetic Recognition (SGR)}
    {A finite, simple, and undirected graph~$G=(V,E)$ and a vertex set~$S \subseteq V(G)$.}
    {Is~$S$ a strong geodetic set of~$G$?}
    \vspace{-.1cm}
    
    In other words, SGR asks whether there exists a path assignment~$I(S)$, as defined in Equation~(\ref{eq:IS}), such that~$\bigcup_{P \in I(S)}{V(P)} = V(G)$.
    
    In this paper we prove the \NP-completeness of SGR and deal with the computational complexity of SGR and SG on some graph classes.
    These results also illustrate that, unlike the classical \textsc{Geodetic} problem, in which geodetic sets can be recognized in polynomial-time (using breadth-first search), the same is not true for strong geodetic sets, unless~$\P=\NP$.\\

	\noindent\textbf{Our results and organization of the paper.} 
    In Section~\ref{sec:Defs}, we introduce additional notation and definitions used in the text.
    We also present some initial considerations about~SG and~SGR.
    
    In Section~\ref{sec:NP-complt}, we prove that~SGR is \NP-complete even for bipartite graphs of bounded diameter and also for bipartite graphs of bounded degree, improving the result of~\cite{multipartite}, that states the \NP-hardness of~SG for general bipartite graphs.
    
    In Section~\ref{sec:co_np}, we prove that SG is \NP-complete for co-bipartite graphs of diameter~2. 
    We also show that both SG and SGR parameterized simultaneously by the diameter and the cardinality of the strong geodetic set are fixed parameter tractable (FPT).
    This result contrasts with the hardness results of Section~\ref{sec:NP-complt} regarding the complexity of SGR parameterized by the max-degree and diameter simultaneously.
        
    In Section~\ref{sec:chordal} we prove that~SG is also NP-complete for chordal graphs of diameter~2, solving an open question posed by Manuel et al.~\cite{manuel2018strong}.
    
    In Section~\ref{sec:polynomial} we present some positive results on solving~SGR for split graphs and for graphs of diameter~2.
    This latter one elucidates the contrast between the complexity of SG (NP-complete) and SGR (polynomial-time solvable).
    Some polynomial-time algorithms for block and cacti graphs are also provided.
    
    We conclude the paper discussing some further research directions in Section~\ref{sec:conclusions}.
\section{Definitions, Notations, and Preliminaries}
\label{sec:Defs}
    
    \subsection{Definitions and Notations}
    
        For a positive integer~$k$, let~$[k] = \{1, 2, \dotsc, n\}$.
        
        In this paper we will only consider simple, connected, and undirected graphs.
        For a graph~$G=(V,E)$ and vertices~t~$u, v \in V(G)$, let~$n = |V(G)|$ and~$m = |E(G)|$.
        We also define~$d(u,v)$ as the \textit{distance} between~$u$ and~$v$, that is, the number of edges in a shortest path (or \textit{geodesic}) between~$u$ and~$v$.
        We will use~$u,v$-\textit{shortest path} to refer to any shortest path between~$u$ and~$v$.
        The \textit{diameter} of~$G$ is the greatest distance between the vertices in~$V(G)$.
        We will denote it as~$diam(G)$.
        
        For a set~$U\subseteq V$, we denote~$G[U]$ as the subgraph of~$G$ \textit{induced by}~$U$.
        We also denote~$N(v)$ as the \textit{neighborhood} (or \textit{open neighborhood}) of~$v$ and~$N[v] = N(v) \cup \{v\}$ as the \textit{closed neighborhood} of~$V$.
        A \textit{simplicial vertex} is one whose neighborhood induces a \textit{clique}, a set of pairwise adjacent vertices.
        The \textit{degree}~$d(v)$ of a vertex~$v$ is the cardinality of its neighborhood.
        Let~$\Delta(G)$ be the maximum degree of the vertices of~$G$.
        
        A \textit{connected component} is a maximal connected induced subgraph of~$G$.
        A vertex~$v$ is a \textit{cut-vertex} of~$G$, if~$G-v$ has more connected components than~$G$.
        A \textit{biconnected subgraph} is one that has no cut-vertices.
        A \textit{biconnected component} is a maximal biconnected subgraph of~$G$.
        Let~$T(u, v)$ be the \textit{interval between~$u$ and~$v$}, that is, the set of vertices belonging to all $u,v$-shortest paths.
        For a set~$S \in V(G)$, let~$T(S) = \bigcup_{u, v \in S} T(u,v)$.
    
        We also use some parameterized complexity concepts.
        See~\cite{downey2012parameterized,DF13,CyganFKLMPPS15} for a complete reference on the subject.
        A parameter is any metric associated with a problem's instance, for example the diameter and the maximum degree.
        A problem~$\Pi$ is \textit{fixed-parameter tractable}, or~\FPT, under the parameter~$k$ if it can be solved by an algorithm~$A$ whose time complexity can be expressed as~$\mathcal{O}\left(f(k) \cdot n^c\right)$, with~$n$ being the size of the input instance, $c$ a positive constant, and~$f(k)$ a computable function.
        The size of the instance includes the size of the parameter.
        In this case~$A$ is called an~\FPT~\textit{algorithm} for~$\Pi$.
        A problem is said to be in~\XP~for a parameter~$k$, if it can be solved in polynomial time when~$k$ is fixed (treated as a constant), that is, it there exists an algorithm for it whose complexity is as~$\mathcal{O}\left(f(k) \cdot n^{g(k)}\right)$, where~$f$ and~$g$ are computable functions.
        
        Another important concept in parameterized complexity is that of \emph{kernelization}.
        A \textit{kernelization algorithm}, or just \emph{kernel}, for a problem~$\Pi$ takes an instance~$I$ and parameter~$k$ and, in time polynomial in~$|I| + k$, outputs an instance~$I'$ with parameter~$k'$, such that~$|I'|, k' \leqslant g(k)$ for some function~$g$.
        Moreover, $(I,k)$ is a \textbf{YES} instance of~$\Pi$ if and only if~$(I',k')$ is a \textbf{YES} instance of~$\Pi$ too.
        The function~$g$ is called the \emph{size} of the kernel and represents a measure of the ``compressibility'' of a problem using polynomial-time preprocessing rules.
        A kernel is called \emph{polynomial} (resp. \emph{linear}) if~$g(k)$ is polynomial (resp. linear) in~$k$.
	    It is nowadays a well-known result in the area that a problem is in~\FPT~if and only if it admits a kernelization algorithm.
	    However, the kernel that one obtains in this way is typically of size at least exponential in the parameter.
	    A natural problem in this context is to find polynomial or linear kernels for problems in \FPT.
	    
	    As in polynomial reductions, a \textit{polynomial parameter transformation}~\cite{BodlaenderTY11} can be used for parameterized problems. 
	    Such transformation from a parameterized problem~$\Pi_1$ (with parameter~$k$) to a parameterized problem~$\Pi_2$ (with parameter~$k'$) is an algorithm that, given an instance~$(x,k)$ of~$\Pi_1$, computes in polynomial time an equivalent instance $(x',k')$ of~$\Pi_2$, such that~$k'$ is polynomially bounded depending only on~$k$.
	    
    \subsection{Preliminaries}
        
        In order to further illustrate the relation between the complexities of~SG and SGR, we state the following proposition that reinforces the intuition that~SG is not computationally easier than~SGR.
        It is easy to see that each simplicial vertex must belong to any strong geodetic set of~$G$.
        \begin{proposition}\label{prop:SGRtoSG}
            \textsc{SGR} is polynomially reducible to \textsc{SG}.
        \end{proposition}
        \begin{proof}
            Let~$\alpha = (G,S)$ be an instance of SGR on a graph~$G = (V, E)$ and~$S \subseteq V(G)$.
            We create an instance~$\beta = (G',k)$ of~SG on the graph~$G' = (V',E')$ and the positive integer~$k$, where~$V' = V(G) \cup \{x_v \mid v \in S\}$, $E' = E(G) \cup \{vx_v \mid v \in S\}$, and~$k = |S|$.
            In other words, we add a pendant vertex~$x_v$ to each vertex~$v$ of~$S$.

            Let~$\alpha$ be a \textbf{YES} instance of SGR.
            Then~$S$ is a strong geodetic set for some shortest path assignment~$I(S)$.
            We state that~$G'$ has a strong geodetic set~$S' = \{x_v \mid v \in S\}$ with a shortest path assignment defined as follows: for each~$u,v$-geodesic~$P \in I(S)$, $I(S')$ contains the path~$(x_u, V(P), x_v)$, which is an $x_u, x_v$-geodesic in~$G'$.
            Therefore~$S'$ is a strong geodetic set of size~$k$ in~$G'$.

            Now, let~$\beta$ be a \textbf{YES} instance of SG.
            Since the~$k$ vertices~$\{x_v \mid v \in S\}$ are simplicial, they compose the strong geodetic set of~$G'$.
            Consequently~$\alpha$ is a \textbf{YES} instance as well, since it is possible to obtain a shortest path assignment~$I(S)$ by adding each~$P \in I(S')$ to~$I(S)$ by removing its endpoints.
        \end{proof}
        
        Proposition~\ref{prop:SGRtoSG} reveals a straightforward manner to solve SGR by solving~SG and also provides a tool to transfer the hardness of~SGR to~SG.
        We can also observe that, when considering the size of~$S$ in~SGR as a parameter, then it is a polynomial parameter transformation as well.
        
        Some upper and lower bounds for the strong geodetic number have been proposed in the literature.
        
\section{\NP-Completeness of \textsc{Strong Geodetic Recognition}}
\label{sec:NP-complt}

    We present a polynomial reduction from an \NP-complete~\cite{schaefer1978complexity} variant of the \textsc{3-SAT} problem, \textsc{3-SAT}$_3$, to SGR.
    An instance of~\textsc{3-SAT}$_3$ consists of a set~$X=\{x_1, x_2, \dotsc, x_n\}$ of variables and a set~$C=\{C_1, C_2, \dotsc, C_m\}$ of clauses, where each clause has~2 or~3 literals (a variable or a negated variable).
    In addition, any variable appears at most~3 times.

    \begin{theorem}\label{strongrec}
        \textsc{Strong Geodetic Recognition} is \NP-complete.
    \end{theorem}
    \begin{proof}
        Let~$G=(V,E)$ and~$S \subseteq V$ compose an instance of SGR.
        The problem is clearly in~\NP, since we can use a shortest path assignment as a certificate to verify in polynomial-time whether all vertices of~$G$ are covered by the specified paths and that each pair $u,v \in S$ has exactly one valid $u,v$-shortest path in that assignment.

        Now, we present a polynomial reduction from~\textsc{3-SAT}$_3$ to SGR.
        Let~$X=\{x_1, x_2, \dotsc, x_n\}$ be the set of variables and~$C=\{c_1, c_2, \dotsc, c_m\}$ be the set of clauses of a~\textsc{3-SAT}$_3$ instance.
        We assume that each variable appears~2 or~3 times in~$C$ and, also, that every variable appears at least once on its positive form and once on its negative form. Otherwise, let $x_i$ be a variable that only appears either on a positive or negative form, we can construct an equivalent instance by removing $x_i$ and the clauses it appears by setting it as true or false, respectively. Given that, each literal can satisfy at most 2 clauses.

        Now we construct an equivalent instance of \textsc{Strong Geodetic Recognition} on a graph $G=(V,E)$ defined as follows (Figure~\ref{fig:SGR} depicts an example of the construction). For each variable~$x_i \in X$ add a gadget containing~8 vertices (variable gadget): $x_i$, $x_i'$, $\overline{x_i}$, $\overline{x_i}'$, $w_i$, $\overline{w_i}$, $p_i$, and~$q_i$.
        Then add the edges~$x_i w_i$, $w_i x_i'$, $\overline{x_i} \overline{w_i}$, $\overline{w_i} \overline{x_i}'$, $q_i \overline{x_i}'$, $q_i x_i'$, $p_i x_i$, and~$p_i \overline{x_i}$.
        
        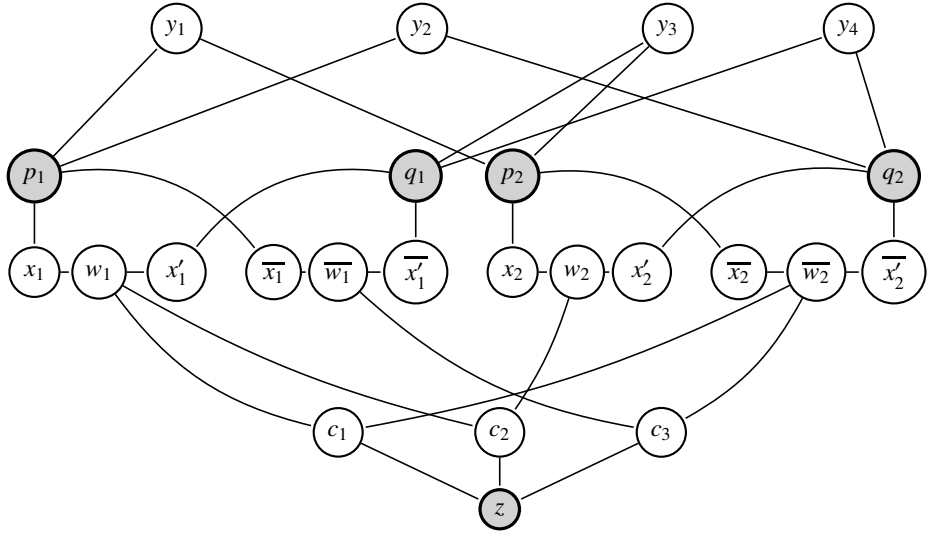
\begin{figure*}[t]
            \centering
            \begin{tikzpicture}[scale=0.85, transform shape,
            > = stealth, 
            shorten > = 1pt, 
            auto,
            node distance = 1.5cm, 
            semithick 
        ]

        \tikzstyle{every state}=[
            draw = black,
            thick,
            fill = white,
            minimum size = 5mm
        ]
		
        \node[state] (p) [very thick,fill = gray!35]{$p_1$};
        \node[state] (a1) [below of=p] {$x_1$};
        \node[state] (w1) [right of=a1, node distance = 1cm] {$w_1$};
        \node[state] (a2) [right of=w1, node distance = 1.2cm] {$x_1'$};
        
        \node[state] (a3) [right of=a2] {$\overline{x_1}$};
        \node[state] (w2) [right of=a3, node distance = 1cm] {$\overline{w_1}$};
        \node[state] (a4) [right of=w2, node distance = 1.2cm] {$\overline{x_1'}$};
        
        \node[state] (q) [very thick,fill = gray!35] [above of=a4] {$q_1$};

        \node[state] (p2) [very thick,right of=q,fill = gray!35]{$p_2$};
        \node[state] (x1) [below of=p2] {$x_2$};
        \node[state] (y1) [right of=x1, node distance = 1cm] {$w_2$};
        \node[state] (x2) [right of=y1, node distance = 1cm] {$x_2'$};
        
        \node[state] (x3) [right of=x2] {$\overline{x_2}$};
        \node[state] (y2) [right of=x3, node distance = 1.2cm] {$\overline{w_2}$};
        \node[state] (x4) [right of=y2, node distance = 1.2cm] {$\overline{x_2'}$};
        
        \node[state] (q2) [very thick,fill = gray!35] [above of=x4] {$q_2$};

        \node[state] (c1) [below of=w2, node distance = 2.5cm] {$c_1$};
        \node[state] (c2) [right of=c1, node distance = 2.5cm] {$c_2$};
        \node[state] (c3) [right of=c2, node distance = 2.5cm] {$c_3$};
        
        \node[state] (z) [very thick ,fill = gray!35] [below of=c2, node distance = 1.2cm] {$z$};
        
        \node[state] (u1) [above of=a2, node distance = 3.8cm] {$y_1$};
        \node[state] (u2) [right of=u1, node distance = 3.8cm] {$y_2$};
        \node[state] (u3) [right of=u2, node distance = 3.8cm] {$y_3$};
        \node[state] (u4) [right of=u3, node distance = 2.8cm] {$y_4$};

        \path[-] (p) edge (a1);
        \path[-] (p) edge [bend left] (a3);
        \path[-] (q) edge [bend right] (a2);
        \path[-] (q) edge (a4);
        
        \path[-] (a1) edge (w1);
        \path[-] (w1) edge (a2);
        \path[-] (a3) edge (w2);
        \path[-] (w2) edge (a4);

        \path[-] (p2) edge (x1);
        \path[-] (p2) edge [bend left] (x3);
        \path[-] (q2) edge [bend right] (x2);
        \path[-] (q2) edge (x4);
        
        \path[-] (x1) edge (y1);
        \path[-] (y1) edge (x2);
        \path[-] (x3) edge (y2);
        \path[-] (y2) edge (x4);

        
        \path[-] (w1) edge [bend right = 20] (c1);
        \path[-] (w1) edge [bend right = 8] (c2);
        \path[-] (w2) edge [bend right = 16] (c3);
        
        \path[-] (y1) edge [bend left = 8] (c2);
        \path[-] (y2) edge [bend left = 8] (c1);
        \path[-] (y2) edge [bend left = 16] (c3);

        \path[-] (c1) edge (z);
        \path[-] (c2) edge (z);
        \path[-] (c3) edge (z);
        
        \path[-] (p) edge (u1);
        \path[-] (p) edge (u2);
        \path[-] (q) edge (u3);
        \path[-] (q) edge (u4);
        \path[-] (u1) edge (p2);
        \path[-] (u2) edge (q2);
        \path[-] (u3) edge (p2);
        \path[-] (u4) edge (q2);
        

    \end{tikzpicture}
            \caption{An instance of SGR arising from an instance of \textsc{3-SAT}$_3$: $X=\{x_1,x_2\}$, $C=\{c_1,c_2,c_3\}$, with $c_1 = (x_1,\overline{x_2})$, $ c_2 = (x_1,x_2)$ and $c_3 = (\overline{x_1},\overline{x_2})$. The vertices marked in gray belong to $S$.}
            \label{fig:SGR}
        \end{figure*}

        For each clause~$c_i \in C$ add a vertex~$c_i$.
        Moreover, add a vertex~$z$ adjacent to all vertices~$c_i \in C$.
        Now, add the edges that represent the relation between variables and clauses as follows: Let~$c_i \in C$ be a clause, then, for each positive literal~$x_i \in c_i$, add the edge~$c_i w_i$, and, for each negative literal~$\overline{x_i} \in c_i$, add the edge~$c_i \overline{w_i}$.
        Repeat this procedure for all clauses in~$C$.

        Finally, for every pair of vertices such as~$(p_i, p_j)$, $(p_i, q_j)$, and~$(q_i, q_j)$, with $i \neq j$, add a new vertex~$y$, an edge between the first vertex of the pair and~$y$ and an edge between~$y$ and the second vertex of the pair.
        Thus, creating a path of size~2 between each pair of vertices as described.
        Hence we obtain~$G$.
        
        Let~$P=\{p_i \mid i \in [n]\}$, $Q=\{q_i \mid i \in [n]\}$, $W=\{w_i, \overline{w_i} \mid i \in [n]\}$, and~$S= P \cup Q \cup \{z\}$.
        The constructed instance consists in deciding whether~$S$ is a strong geodetic set of~$G$.

        Now we prove that if the instance of~\textsc{3-SAT}$_3$ is satisfiable, then~$S$ is a strong geodetic set of~$G$.
        Let~$T$ be a truth assignment of~$X$ satisfying all clauses of~$C$.
        At first, note that the length of a shortest path from~$z$ and a vertex in~$P \cup Q$ is~4.
        So, if~$x_i$ is set to true at $T$, then assign~$(p_i,x_i,w_i,c,z)$, a~$p_i,z$-shortest path, and~$(q_i,x_i',w_i,c',z)$, a~$q_i, z$-shortest path, with~$c$ and~$c'$ denoting the clauses that~$x_i$ satisfies when set to \textit{true}.
        Observe that any literal satisfies either one or two clauses.
        Thus, if two clauses are satisfied, then~$c \neq c'$, otherwise, $c = c'$.
        Now, assign the shortest path~$(p_i,\overline{x_i},\overline{w_i},\overline{x_i}',q_i)$ between~$p_i$ and~$q_i$.
        Note that~$D(p_i,q_i) = 4$.

        If~$x_i$ is set to false in~$T$, then the paths will be chosen on an analogous way.
        We will choose the paths~$(p_i,\overline{x_i},\overline{w_i},c,z)$, $(q_i,\overline{x_i}',\overline{w_i},c',z)$, and~$(p_i,x_i,w_i,x_i',q_i)$.
        By this time, all vertices in variable gadgets and all clause vertices are covered.
        This holds because the vertices~$w_i$ and~$\overline{w_i}$ are adjacent to all clause vertices, each one satisfied, and it is possible to cover these clause vertices with~$(p_i,z)$ and~$(q_i,z)$-shortest paths.
        It remains to define the paths between vertices in~$S \setminus \{z\}$ that are in different variable gadgets.
        We assign to~$I(S)$ the unique length~2 shortest path between such vertices.
        Finally, note that all vertices are covered, hence, $S$ is a strong geodetic set of~$G$.

        Now, assume that~$S$ is a strong geodetic set of~$G$. Consider the variable $x_i \in X$ and observe that one of the following options holds:
        \begin{itemize}
            \item The~$p_i,z$-shortest path passes through~$x_i$ and the~$q_i,z$-shortest path passes through~$x_i'$.
            \item The~$p_i,z$-shortest path passes through~$\overline{x_i}$ and the~$q_i,z$-shortest path passes through~$\overline{x_i}'$.
        \end{itemize}
        This claim holds because, otherwise, one of the vertices in~$\{x_i,x_i',\overline{x_i},\overline{x_i}'\}$ would not be covered, since a~$p_i,q_i$-shortest path can cover either~$x_i$ and~$x_i'$ or~$\overline{x_i}$ and~$\overline{x_i}'$.

        Therefore, the variable gadget forces a choice between either a positive or a negative literal.
        It is important to note that only shortest paths between~$z$ and a vertex in~$P \cup Q$ are able to cover clause vertices.
        Now, consider the following truth assignment for~$X$: for each~$x_i \in X$, if~$I(S)$ assigns the~$(p_i,z)$ and~$(q_i,z)$-shortest paths to pass through~$x_i$ and~$x_i'$, then set~$x_i$ to true, otherwise, set~$x_i$ to false.
        This constructed truth assignment satisfies all clauses, since~$S$ is a strong geodetic set of~$G$, which must cover all clause vertices.
        Hence, the~\textsc{3-SAT}$_3$ instance is satisfiable and the proof is concluded.
    \end{proof}

    \begin{corollary}\label{bipartido}
        \textsc{Strong Geodetic Recognition} is \NP-complete even when restricted to bipartite graphs with diameter bounded by~6.
    \end{corollary}
    \begin{proof}
        Consider the graph~$G=(V,E)$ constructed on Theorem~\ref{strongrec}.
        Let~$U = \{x_i,x_i',\overline{x_i},\overline{x_i}'\}$ and let~$Y$ be the set containing all~$y$ vertices of~$G$.
        Now, let~$A = P \cup Q \cup W \cup \{z\}$ and~$B = C \cup U \cup Y$.
        Note that both $A$ and~$B$ are independent sets, hence, $G$ is bipartite.
        Also, observe that the largest distance in the graph occurs between a vertex $y \in Y$ and a clause vertex that is not satisfied by either variable gadgets adjacent to~$y$, and this distance is~6.
    \end{proof}

    This result indicates that, unless \P=\NP, SGR parameterized by the diameter is not in~\XP, as the problem remains \NP-complete even for graphs with bounded diameter.
    Next we prove that the same holds for graphs of bounded degree, where we adapt the reduction presented on Theorem~\ref{strongrec}.
    
    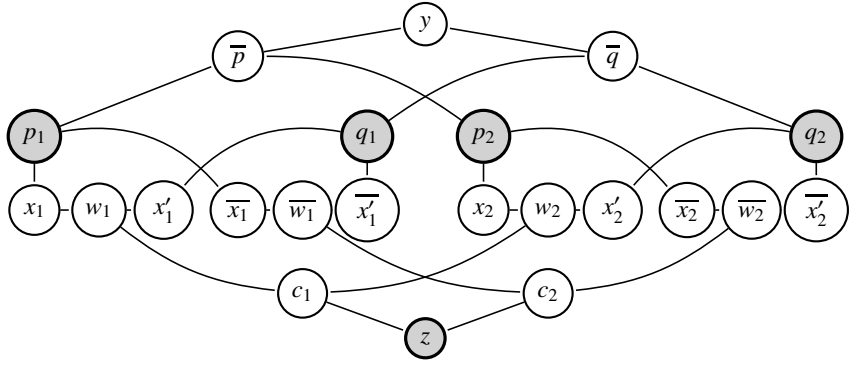
\begin{figure*}[t]
        \centering
        \begin{tikzpicture}[scale=0.85, transform shape,
            > = stealth, 
            shorten > = 1pt, 
            auto,
            node distance = 1.15cm, 
            semithick 
        ]

        \tikzstyle{every state}=[
            draw = black,
            thick,
            fill = white,
            minimum size = 6mm
        ]
		
        \node[state] (p) [very thick,fill = gray!35]{$p_1$};
        \node[state] (a1) [below of=p] {$x_1$};
        \node[state] (w1) [right of=a1, node distance = 1cm] {$w_1$};
        \node[state] (a2) [right of=w1, node distance = 1cm] {$x_1'$};
        
        \node[state] (a3) [right of=a2] {$\overline{x_1}$};
        \node[state] (w2) [right of=a3, node distance = 1cm] {$\overline{w_1}$};
        \node[state] (a4) [right of=w2, node distance = 1cm] {$\overline{x_1'}$};
        \node[state, draw=none] (center) [right of=a4, node distance = 0.9cm] {};
        
        \node[state] (q) [very thick,fill = gray!35] [above of=a4] {$q_1$};

        \node[state] (p2) [very thick,right of=q,fill = gray!35, node distance = 1.8cm]{$p_2$};
        \node[state] (x1) [below of=p2] {$x_2$};
        \node[state] (y1) [right of=x1, node distance = 1cm] {$w_2$};
        \node[state] (x2) [right of=y1, node distance = 1cm] {$x_2'$};
        
        \node[state] (x3) [right of=x2] {$\overline{x_2}$};
        \node[state] (y2) [right of=x3, node distance = 1cm] {$\overline{w_2}$};
        \node[state] (x4) [right of=y2, node distance = 1cm] {$\overline{x_2'}$};
        
        \node[state] (q2) [very thick,fill = gray!35] [above of=x4] {$q_2$};

        \node[state] (c1) [below of=w2, node distance = 1.3cm] {$c_1$};
        \node[state] (c2) [below of=y1, node distance = 1.3cm] {$c_2$};

        \node[state] (z) [very thick ,fill = gray!35] [below of=center, node distance = 2cm] {$z$};
        
        \node[state] (yy1) [above of=a3, node distance = 2.4cm] {$\overline{p}$};
        \node[state] (yy2) [above of=x2, node distance = 2.4cm] {$\overline{q}$};
        
        \node[state] (y) [above of=center, node distance = 2.9cm] {$y$};

        \path[-] (p) edge (a1);
        \path[-] (p) edge [bend left] (a3);
        \path[-] (q) edge [bend right] (a2);
        \path[-] (q) edge (a4);
        
        \path[-] (a1) edge (w1);
        \path[-] (w1) edge (a2);
        \path[-] (a3) edge (w2);
        \path[-] (w2) edge (a4);

        \path[-] (p2) edge (x1);
        \path[-] (p2) edge [bend left] (x3);
        \path[-] (q2) edge [bend right] (x2);
        \path[-] (q2) edge (x4);
        
        \path[-] (x1) edge (y1);
        \path[-] (y1) edge (x2);
        \path[-] (x3) edge (y2);
        \path[-] (y2) edge (x4);

        
        \path[-] (w1) edge [bend right = 16] (c1);
        \path[-] (y1) edge [bend left = 16] (c1);
        \path[-] (w2) edge [bend right = 16] (c2);
        \path[-] (y2) edge [bend left = 16] (c2);

        \path[-] (c1) edge (z);
        \path[-] (c2) edge (z);
        
        \path[-] (p) edge (yy1);
        \path[-] (p2) edge [bend right = 18] (yy1);
        \path[-] (q) edge [bend left = 18] (yy2);
        \path[-] (q2) edge (yy2);
        \path[-] (yy1) edge (y);
        \path[-] (yy2) edge (y);


    \end{tikzpicture}
        \caption{An instance of SGR arising from an instance of \textsc{3-SAT}$_3$: $X=\{x_1,x_2\}$, $C=\{c_1,c_2\}$, with $c_1 = (x_1,x_2)$, $ c_2 = (\overline{x_1},\overline{x_2})$. The vertices marked in gray belong to $S$.}
        \label{fig:SGRmaxdegree}
    \end{figure*}

    \begin{theorem}\label{strongrecdegree}
        \textsc{Strong Geodetic Recognition} restricted to bipartite graphs with maximum degree bounded by~4 is \NP-complete.
    \end{theorem}
    \begin{proof}
        We will reduce an instance~$\Pi$ of \textsc{3-SAT}$_3$, as in Theorem \ref{strongrec}, to an instance~$\Pi'$ of SGR.
        Let~$X=\{x_1, x_2, \dotsc, x_n\}$ be the set of variables and~$C=\{c_1, c_2, \dotsc, c_m\}$ be the set of clauses of~$\Pi$.
        We assume that~$|X|$ and~$|C|$ are exact powers of~2, since, for any instance $\Pi$, artificial variables and clauses can be added to $\Pi$ in order to satisfy this assumption, making the equivalent resulting instance at most twice as large as $\Pi$.

        Let~$G=(V,E)$ be the graph associated with the instance of SGR obtained in Theorem~\ref{strongrec}.
        Now, we present some adaptations on~$G$ in order to construct a graph~$G'$ associated with~$\Pi'$.
        Figure \ref{fig:SGRmaxdegree} depicts an example of the construction.
        Let~$G' := G[P \cup Q \cup U \cup W \cup C]$ and then do the following modifications to~$G'$: add a vertex~$z$ and connect~$z$ to all clause vertices by using a full binary tree~$T_z$, that is rooted at~$z$ and whose leaves are all the clause vertices.
        Observe that this operation results on~$z$ as a~2-degree vertex and all introduced auxiliary vertices have degree~3.
        Moreover, the size of~$T_z$ is the number of clause vertices of $G$.

        Remembering that~$P = \{p_1, p_2, \dotsc, p_n\}$ and~$Q = \{q_1, q_2, \dotsc, q_n\}$, we add some more gadgets to~$G'$ as follows.
        Add a vertex~$\overline{p}$ and connect it to all vertices in~$P$ by using a binary tree~$T_P$, as explained for~$T_z$.
        Analogously, add a vertex~$\overline{q}$ and connect it to all vertices in~$Q$ by using an additional binary tree~$T_Q$.
        Finally, add a vertex~$y$ and the edges~$\overline{p}y$ and~$\overline{q}y$, resulting in a binary tree~$T_y$.
        
        Concluding the construction, let~$\alpha = \log_2{n}$. If~$\alpha > 1$, then, for every edge~$e$ among the edges~$\{p_ix_i$, $p_i\overline{x_i}$, $q_i\overline{x_i}',q_ix_i'\}$, for every~$i \in [n]$, replace~$e$ by a path~$P_e$ having~$\alpha$ edges.
        The construction of~$G'$ is complete and now it remains to prove that~$\Pi$ is equivalent to recognizing whether the set~$S = P \cup Q \cup \{z\}$ is a strong geodetic set of~$G'$ (instance~$\Pi'$).
        Observe that~$G'$ is a bipartite graph of maximum degree~4 (a clause vertex associated with a 3-sized clause has~3 literal vertex neighbors and one in~$T_z$).

        Assume that~$\Pi$ is satisfiable, then there exists a truth assignment~$T$ of~$X$ satisfying all clauses.
        Now, we construct a shortest path assignment~$I(S)$ proving that~$S$ is a strong geodetic set of~$G'$.
        For every variable~$x_i$ that is set to true in~$T$, do the following:
        \begin{itemize}
            \item Let~$C_{p_i,z} = (P_{p_i,x_i},w_i,c,P_{c,z})$ be a shortest path between~$p_i$ and~$z$, such that~$P_{p_i,x_i}$ denotes the path that replaces the edge~$p_ix_i$ and~$P_{c,z}$ denotes the unique shortest path between a clause vertex~$c \in N(w_i)$ and~$z$ in~$G'$.
            Observe that~$P_{c,z}$ traverse~$T_z$.
            We add~$C_{p_i,z}$ to~$I(S)$.
            \item Analogously, let~$C_{q_i,z}$ be a shortest path between~$q_i$ and~$z$ such that~$C_{q_i,z} = (P_{q_i,x_i'},w_i,c',P_{c',z})$.
            Here, if~$w_i$ is adjacent to~2 different clause vertices, then~$c' \in N(w_i)$ and~$c' \neq c$, otherwise, $c = c'$.
            We add~$C_{q_i,z}$ to~$I(S)$.
            \item Finally, let~$C_{p_i,q_i}$ be a shortest path between~$p_i$ and~$q_i$ such that~$C_{p_i,q_i} = (P_{p_i,\overline{x_i}},\overline{w_i},\overline{x_i}',P_{\overline{x_i}',q_i})$.
            We add~$C_{p_i,q_i}$ to~$I(S)$.
        \end{itemize}
        Variables that are set to false will be treated analogously, as in Theorem~\ref{strongrec}.
        Now, observe that all vertices in variable gadgets are covered.
        Moreover, since~$\Pi$ is satisfiable, all clause vertices are covered as well, because the shortest path assignment explained covers (satisfies) the same clause vertices (clauses) as the truth assignment~$T$.
        This also implies that all internal vertices of~$T_z$ are covered.
        It remains to determine the shortest paths between vertices in~$S \setminus \{z\}$ lying in different variable gadgets.
        Every such paths will traverse~$T_y$, covering all internal vertices in it.
        Finally, all vertices of~$G'$ are covered and~$S$ is a strong geodetic set of~$G'$.
    
        For the converse, assume that~$S$ is a strong geodetic set of~$G'$, hence, there exists a shortest path assignment~$I(S)$ that covers all vertices of~$G'$.
        First, note that for every variable~$x_i \in X$, both shortest paths between~$p_i$ and~$z$ and between~$q_i$ and~$z$ must traverse either~$w_i$ or~$\overline{w_i}$, in the same way as in Theorem~\ref{strongrec}.
        Moreover, observe that shortest paths between vertices in~$S \setminus \{z\}$ from different variable gadgets will always traverse~$T_y$, assuring that these paths do not cover clause vertices.
        Concluding, since variable gadgets force a choice between a positive or a negative literal, the existence of a shortest path assignment covering all clause vertices indicates the existence of a truth assignment for~$\Pi$ satisfying all clauses, and the proof is concluded.
    \end{proof}

    Observe that the previous result indicates that, unless \P=\NP, SGR parameterized by the maximum degree is not in~\XP, as the problem is \NP-complete even for graphs with max-degree bounded by 4.
    Moreover, by Proposition~\ref{prop:SGRtoSG}, it is possible to conclude that SG restricted to bipartite graphs of maximum degree~4 is also \NP-complete, as the constructed instance~$\Pi'$ is equivalent to an instance~$\phi = (\overline{G},k)$ of~SG, where~$\overline{G}$ is obtained from~$G'$ by adding pendant vertices adjacent to the vertices in~$S$ on~$\Pi'$ and~$k = 2n + 1$. Observe that the max-degree of $\overline{G}$ does not exceed 4.
\section{\textsc{Strong Geodetic} for Co-Bipartite Graphs}
\label{sec:co_np}
    A \textit{co-bipartite} graph is the complement of a bipartite graph.
    Alternatively, a graph is said to be co-bipartite if its vertex set can be partitioned into two cliques.
    Note that the maximum diameter of a connected co-bipartite graph is~3.

    We prove that SG is \NP-complete even for co-bipartite graphs by a polynomial reduction inspired by~\cite{multipartite}. We reduce from the \textsc{dominating set} problem for connected bipartite graphs. Note that SG is in \NP, since one can verify in polynomial time whether a shortest path assignment $I(S)$ used as a certificate is valid, covers all vertices of the graph and has $|S| \leq k$.

    \begin{theorem}\label{theorem_cobipartite}
        \textsc{Strong Geodetic} restricted to co-bipartite graphs is \NP-complete.
    \end{theorem}
    \begin{proof}
        Let~$G=(V,E)$ be a connected bipartite graph whose parts are $A = \{p_1, p_2, \dotsc, p_{|A|}\}$ and~$B = \{q_1, q_2, \dotsc, q_{|B|}\}$ having cardinality at least~2.
        We construct the graph~$H=(V',E')$, with: $V' = V \cup \overline{A} \cup \overline{B} \cup \{a',b'\}$, such that~$\overline{A} = \{a_1, a_2, \dotsc, a_{|A|}\}$ and~$\overline{B} = \{b_1, b_2, \dotsc, b_{|B|}\}$.
        
        The edge set~$E'$ contains all edges in~$E$ plus the necessary additions such that~$a'$ and~$b'$ are universal vertices of~$G'$ and~$A \cup \overline{A} \cup \{a'\}$ and~$B \cup \overline{B} \cup \{b'\}$ are cliques.
        Observe that~$H$ is a co-bipartite graph whose diameter is~2.
        
        Let~$D$ be a dominating set of~$G$, with~$|D| = k$.
        We will show that~$H$ has a strong geodetic set~$S = D \cup \overline{A} \cup \overline{B}$.
        We construct a suitable~$I(S)$ covering all vertices.
        The shortest paths~$(b_1,a',a_1)$ and~$(b_2,b',a_1)$ are assigned to cover~$a'$ and~$b'$, respectively.
        For any vertex~$p_i \in A \setminus S$, it holds that~$p_i$ has at least a neighbor~$u \in B \cap S$, then the shortest path~$(u,p_i,a_i)$ is assigned to cover~$p_i$.
        Finally, for any vertex~$q_i \in B \setminus D$, it holds that~$q_i$ has at least a neighbor~$v \in A \cap S$, so we assign the~$(v,q_i,b_i)$ shortest path to cover~$q_i$.
        Concluding, $S$ is a strong geodetic set of~$H$, with~$|S| = |D| + |V|$.
        
        It remains to prove that if~$S$ is a strong geodetic set of~$H$, with~$|S| \leq k + |\overline{A} \cup \overline{B}|$, then~$G$ has a dominating set~$D$ with~$|D| \leq k$.
        Note that if~$S$ is a strong geodetic set of~$H$, then~$\overline{A} \cup \overline{B} \subseteq S$, since~$\overline{A}$ and~$\overline{B}$ contain only simplicial vertices.
        Now, we show that~$S \cap V$ is a dominating set of $G$.
        Since~$S$ is a strong geodetic set of~$H$, for each vertex~$x \in A \setminus S$ there exists a shortest path in~$I(S)$ that contains~$x$.
        Note that~$D(u,v) \leq 2$, for all~$u,v \in V'$, so there exists a shortest path~$(a,x,b)$ in~$H$ such that~$a,b \in S$.
        Recall that one of the vertices at the shortest path must be in~$B$, and we denote it~$b$.
        This holds because~$x$ has no neighbors in~$\overline{B}$ and~$b'$ cannot be in a shortest path, because~$b'$ is universal.
        Concluding, every vertex~$x \in A \setminus S$ has a neighbor in~$B$ belonging to~$S$ and the same holds for any vertex~$x' \in B \setminus S$.
        Thereafter, $S \cap V$ is a dominating set of~$G$, with~$|S \cap V| \leq k$, since~$|S \cap (V' \setminus V)| = |V|$.
    \end{proof}

    Observe that this result indicates that, unless \P=\NP, SG parameterized by the diameter is not in~\XP, as the problem is \NP-complete even for graphs of diameter~2.
    Nevertheless, the next result proves that SG parameterized by the diameter and the natural parameter~$k$ in conjunct belongs to~\FPT.

    \begin{theorem}\label{fptdiameter}
        Let~$G=(V,E)$ be a graph of diameter~$D$.
        The problem of deciding whether~$G$ has a strong geodetic set of cardinality~$k$ is fixed parameter tractable on the parameters~$D$ and~$k$ in conjunct.
    \end{theorem}
    \begin{proof}
        Let~$S \subseteq V$ with~$|S| = k$ and let~$U = V \setminus S$.
        If~$S$ is a strong geodetic set of~$G$, then every vertex in~$U$ must be internal of some shortest path between vertices in~$S$.
        In addition, observe that there are~$\binom{k}{2}$ pairs of vertices of~$S$, and for each one of these pairs it will be assigned a shortest path that will cover at most~$D - 1$ vertices in~$U$.
        Therefore, if $|V| - k > \binom{k}{2} \times (D - 1)$,
        then no set~$S$ with cardinality~$k$ can be a strong geodetic set.
        Otherwise, $|V| \leq \binom{k}{2} \times (D - 1) + k$,
        resulting that the size of the graph is bounded by a polynomial function on~$D$ and~$k$.
        It follows that we found a polynomial kernel of the problem in polynomial time, and the theorem follows.
    \end{proof}
    
    The idea behind this result is that the graph associated with any \textbf{YES} instance has a limited number of vertices: $ |V| \leq \binom{k}{2} \times (D - 1) + k$. Otherwise, one can assure that it consists of a \textbf{NO} instance. Observe that~SGR is also~\FPT~on the parameters~$D$ and~$k$, with~$k$ indicating the size of the set~$S$ given as input. The same argumentation applies.
\section{The Strong Geodetic Problem for Chordal Graphs}
\label{sec:chordal}

    A graph is said to be \emph{chordal} if it has no induced cycles of length at least~4.
    A \emph{split graph} is that which can be partitioned into a clique and an independent set.
    In this section we present a reduction from the \textsc{Dominating Set} problem for connected split graphs~\cite{bertossi1984dominating} to~SG for chordal graphs.

    \begin{theorem}\label{teorema_chordal}
        The \textsc{Strong Geodetic} problem for chordal graphs with diameter~2 is \NP-complete.
    \end{theorem}
    \begin{proof}
        Let~$G=(V,E)$ be a connected split graph with vertex set partitioned into a clique~$C$ and~$I$ an independent set.
        And let~$H$ be the graph obtained from~$G$ as follows.
        For each vertex~$u \in I$ add the vertex~$x_u$ to~$H$, and, for each vertex~$v \in C$ add the vertex~$y_v$ to~$H$, and finally, add a universal vertex~$z$ to~$H$.
        Moreover, for each vertex~$u \in I$ add an edge between~$u$ and~$x_u$, and, for each vertex~$v \in C$ add an edge between~$v$ and~$y_v$.
        Observe that~$H$ is a chordal graph of diameter~2.
        
        Assume that~$G$ has a dominating set~$D$, with~$|D| \leq k$.
        Let~$X = \{x_u \mid u \in I\}$ and~$Y = \{y_u \mid u \in C\}$.
        We show that~$S = D \cup X \cup Y$ is a strong geodetic set of~$H$.
        First, note that any $y,y'$-shortest path contains~$z$, with~$y,y' \in Y$ and~$y \neq y'$.
        Then we include the shortest path~$(y,z,y')$ in~$I(S)$.
        Now, let~$u$ be a vertex in~$I \setminus S$, which implies that~$u \notin D$, and, then~$u$ has a neighbor~$v \in C \cap D$, that is, $v \in S$.
        We include the~$x_u,v$-shortest path that contains~$u$ in~$I(S)$.
        Let~$p$ be a vertex in~$C \setminus S$.
        Analogously, $p$ has a neighbor~$q \in I \cap S$.
        We include the $y_p,q$-shortest path that contains~$p$ in~$I(S)$.
        Thus~$S$ is a strong geodetic set of~$H$, with~$|S| \leq k + |V|$.
        
        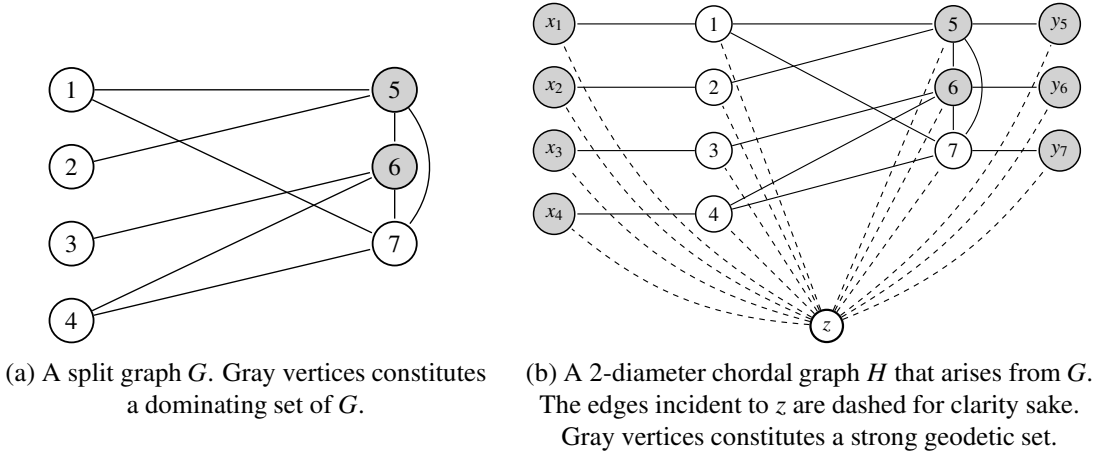
\begin{figure*}[t!]
        \centering
        
        \begin{subfigure}[t]{0.43\linewidth}
            \centering
        	\scalebox{0.85}{\begin{tikzpicture}[
            > = stealth, 
            shorten > = 1pt, 
            auto,
            node distance = 1.2cm, 
            semithick 
        ]

        \tikzstyle{every state}=[
            draw = black,
            thick,
            fill = white,
            minimum size = 1mm
        ]
		
        \node[state] (1) {$1$};
        \node[state] (2) [below of=1] {$2$};
        \node[state] (3) [below of=2] {$3$};
        \node[state] (4) [below of=3] {$4$};

        \node[state] (5) [right of = 1, node distance = 5cm, fill = gray!35] {$5$};
        \node[state] (6) [below of=5, fill = gray!35] {$6$};
        \node[state] (7) [below of=6] {$7$};

        \path[-] (5) edge (6);
        \path[-] (6) edge (7);
        \path[-] (5) edge [bend left = 37] (7);

        \path[-] (1) edge (5);
        \path[-] (1) edge (7);
        \path[-] (2) edge (5);
        \path[-] (3) edge (6);
        \path[-] (4) edge (6);
        \path[-] (4) edge (7);


    \end{tikzpicture}}
            \caption{A split graph $G$. Gray vertices constitutes a dominating set of $G$.}
          \label{fig:split}
        \end{subfigure}
        ~
        \begin{subfigure}[t]{0.52\linewidth}
            \centering
            \scalebox{0.7}{\begin{tikzpicture}[
            > = stealth, 
            shorten > = 1pt, 
            auto,
            node distance = 1.2cm, 
            semithick 
        ]

        \tikzstyle{every state}=[
            draw = black,
            thick,
            fill = white,
            minimum size = 1mm
        ]
		
        \node[state] (1) {$1$};
        \node[state] (2) [below of=1] {$2$};
        \node[state] (3) [below of=2] {$3$};
        \node[state] (4) [below of=3] {$4$};
        
        \node[state] (x1) [left of=1,node distance=3cm, fill = gray!35] {$x_1$};
        \node[state] (x2) [below of=x1, fill = gray!35] {$x_2$};
        \node[state] (x3) [below of=x2, fill = gray!35] {$x_3$};
        \node[state] (x4) [below of=x3, fill = gray!35] {$x_4$};

        \node[state] (5) [right of = 1, node distance = 4.5cm, fill = gray!35] {$5$};
        \node[state] (6) [below of=5, fill = gray!35] {$6$};
        \node[state] (7) [below of=6] {$7$};
        
        \node[state] (y5) [right of = 5, node distance = 2cm, fill = gray!35] {$y_5$};
        \node[state] (y6) [below of=y5, fill = gray!35] {$y_6$};
        \node[state] (y7) [below of=y6, fill = gray!35] {$y_7$};
        
        \node[state] (z) [very thick,below right of=4, node distance = 3cm] {$z$};

        \path[-] (5) edge (6);
        \path[-] (6) edge (7);
        \path[-] (5) edge [bend left = 37] (7);
        
        \path[-] (1) edge (x1);
        \path[-] (2) edge (x2);
        \path[-] (3) edge (x3);
        \path[-] (4) edge (x4);
        \path[-] (5) edge (y5);
        \path[-] (6) edge (y6);
        \path[-] (7) edge (y7);
        
        \path[-] (1) edge (5);
        \path[-] (1) edge (7);
        \path[-] (2) edge (5);
        \path[-] (3) edge (6);
        \path[-] (4) edge (6);
        \path[-] (4) edge (7);
        
        \path[-] (z) edge [dashed] (1);
        \path[-] (z) edge [dashed] (2);
        \path[-] (z) edge [dashed] (3);
        \path[-] (z) edge [dashed] (4);
        \path[-] (z) edge [dashed] (5);
        \path[-] (z) edge [dashed] (6);
        \path[-] (z) edge [dashed] (7);
        \path[-] (z) edge [dashed, bend left = 18] (x1);
        \path[-] (z) edge [dashed, bend left = 18] (x2);
        \path[-] (z) edge [dashed, bend left = 18] (x3);
        \path[-] (z) edge [dashed, bend left = 18] (x4);
        \path[-] (z) edge [dashed, bend right = 15] (y5);
        \path[-] (z) edge [dashed, bend right = 15] (y6);
        \path[-] (z) edge [dashed, bend right = 15] (y7);


    \end{tikzpicture}}
            \caption{A 2-diameter chordal graph $H$ that arises from $G$. The edges incident to $z$ are dashed for clarity sake. Gray vertices constitutes a strong geodetic set.}
          \label{fig:chordal}
        \end{subfigure}
        \caption{Example illustrating the polynomial reduction presented on Theorem~\ref{teorema_chordal}}
        \label{fig:redchordal}
    \end{figure*}
        
        For the converse, assume that~$H$ has a strong geodetic set~$S$ with~$|S| \leq k + |V|$.
        First, observe that all vertices in~$X \cup Y$ are simplicial, thus~$X \cup Y \subseteq S$.
        Note that if some strong geodetic set~$S$ of~$H$ contains~$z$, then~$S \setminus \{z\}$ is also a strong geodetic set. Hence, we will assume that~$z \notin S$. 
        
        We now prove that~$D = S \cap V$ is a dominating set of~$G$.
        Let~$u \in V \setminus D$, there exists an~$m,n$-shortest path in~$H$ that contains $u$.
        As the diameter of~$H$ is~2, this path must be in the form~$(m,u,n)$.
        Suppose for a contradiction that neither~$m$ or~$n$ are in~$V$.
        So there are three cases for~$m$ and~$n$: $m$ and~$n$ are in~$X$, $m \in X$ and~$n \in Y$, and~$m$ and~$n$ are in~$Y$.
        For all the cases there would be a unique~$m,n$-shortest path: $(m,z,n)$, that leads to a contradiction. Therefore, $m$ or~$n$ must be in~$V \cap S$, so~$u$ has a neighbor in~$D$.
        Hence, $D$ is a dominating set of~$G$, and~$|D| \leq k$, since~$|(X \cup Y) \cap S| = |V|$.
    \end{proof}
\section{Polynomial Instances of SG and SGR}
\label{sec:polynomial}

    In this section we present some positive results.

    Despite the \NP-completeness of~SG for chordal graphs as seen, we can prove that the problem can be solved in linear time for \textit{block graphs}, a subclass of chordal graphs. A block graph is one in which all biconnected components are complete subgraphs. 

    A block graph is one in which all biconnected components are complete subgraphs. Now, we introduce the definition of a cut-tree, which is an important structure to understand block graphs and the next result.

    \begin{definition}[Cut-tree]\label{def_block}
        A cut-tree $T=(V',E')$ of a graph $G$ is a tree in which each vertex represents a biconnected component or a cut-vertex of $G$. There is an edge~$e \in E'$ for each pair of a cut-vertex $a$ and a biconnected component~$C$ of $G$, such that~$a \in C$.
    \end{definition}

    \begin{theorem}\label{teorema_block}
        Let~$G=(V,E$ be a block graph. The set~$S$ of all simplicial vertices of a block graph~$G$ is the minimum strong geodetic set of~$G$.
    \end{theorem}
    \begin{proof}
     It holds that $S$ must be contained in any strong geodetic set of $G$, so if we prove that $S$ is a strong geodetic set, it has minimum cardinality. The vertices of the graph can be partitioned into two sets: simplicial vertices $S$ and cut-vertices $A$. Let $T = (V', E')$ be a cut-tree of $G$ and~$v \in A$. Consider $C_1$ and~$C_2$ as two connected components of $T[V' \setminus \{v\}]$. Let $f_1$ be a leaf of $T[C_1]$ and~$f_2$ a leaf of $T[C_2]$. Note that both $f_1$ and $f_2$ represent biconnected components of~$G$, which are complete graphs. As a result, each connected component denoted by~$f_1$ and $f_2$ has at least one simplicial vertex: $s_1$ and $s_2$, respectively. Finally, the $s_1,s_2$-shortest path contains $v$, whereas it is a cut-vertex. Thereafter, $S$ is a minimum strong geodetic set of $G$.
    \end{proof}
    
    \begin{corollary}\label{teorema_block_P}
        There is a linear-time algorithm that solves the \textsc{strong geodetic} problem for block graphs.
    \end{corollary}
    \begin{proof}
        The algorithm consists in running a depth first search to find the set $A$ of cut-vertices of the graph. Then, return $V \setminus A$ as solution.
    \end{proof}
    
    \begin{corollary}
        There is a linear-time algorithm that solves the \textsc{strong geodetic recognition} problem for block graphs.
    \end{corollary}
    \begin{proof}
        Given any set $X \subseteq V$, if $S \subseteq X$ then $X$ is a strong geodetic set, otherwise~$X$ is not a strong geodetic set.
    \end{proof}
    
    We also obtained a polynomial-time algorithm for~SG on \textit{cacti graphs}.
    A cactus graph is a connected one in which every edge belongs to at most one simple cycle.
    In the proof, we consider a down-top approach on a \textit{cut tree}~$T$ representation of the given cactus~$G$.
    The main idea is to guarantee that all vertices of each biconnected component are covered, and this can be achieved considering one at a time, after dealing with some technicalities. The procedure constructs a minimum strong geodetic set by adding the minimum amount of vertices required for each biconnected component.
    Including the simplicial vertices, for each leaf~$\ell$ of~$T$, we show that at most two vertices of~$\ell$ are required to compose the optimal solution, where we test the parity of the cycle of~$\ell$.
    Moreover, for each cycle~$C$ represented by an internal vertex~$v$ of~$T$ we consider the distance between the farther cut vertices in~$C$, where we prove that no vertices between such vertices are in the optimal solution and at most one vertex of~$C$ must be included, depending on the parity of the size of~$C$.
%
%
    
    We will first illustrate a pre-processing procedure. The procedure receives a cactus graph and its cut-tree. We will consider that the received cut-tree has at least two nodes, as otherwise the algorithm simple consists in solving the SG for a cycle or an~edge.

\vbox{
\begin{enumerate}
    \item 
    \textbf{Input:} A cactus graph $G = (V,E)$ and its cut-tree $T = (V', E')$.
    \item
    Initialize $S$ as an empty set.
    \item
    For each leaf $\ell$ in $ T $ do:
    \subitem
    - If $\ell$ corresponds to an edge $uv$ of $G$ (a biconnected component that is an edge), then add its simplicial vertex to $S$.
    \subitem
    - If $\ell$ corresponds to an even cycle $C$ of length $l$ whose cut-vertex is $a$, add a vertex $v \in C$ to $S$ such that $D(a,v) = \dfrac{l}{2}$.
    \subitem
    - If $\ell$ corresponds to an odd cycle $C$ of length $ l $ whose cut-vertex is $a$, add two vertices $u, v \in C$ to $S$, such that $D(a,u) = D(a,v) = \left\lfloor \dfrac{l}{2} \right\rfloor$.
    \item
    Finish pre-processing.
\end{enumerate}
}

Having finished pre-processing, we now define how to process each biconnected component (block) associated to internal vertices of $T$. Let $t$ be an internal vertex of~$ T $, if $t$ represents an odd cycle $C$ of length $ l $ do: Define $A$ as the set of cut-vertices of~$G$ present in $C$. Consider $x_1, x_k \in A$ with $ x_1 \neq x_k $ and $P = (x_1, x_2, \dotsc, x_{k-1}, x_k)$ as the longest path between $x_1 $ and $x_k$ in $ C $. Let $j = \left\lfloor \dfrac{1 + k}{2} \right\rfloor$ and $v = x_j$. If $\bigcup_{p,q \in A} p(p,q) \neq V(C)$ add $v$ to $S$, otherwise, proceed to the next block. Here, $p(p,q)$ denotes the unique shortest path between $p$ and $q$ in $C$.

If $t$ represents an even cycle $C$ of length $l$ in $G$ do: Define $A$ as the set of cut-vertices of $G$ contained in $C$. If there are $a_1, a_2 \in A$ such that $ D(a_1, a_2) = \frac{l}{2}$ proceed to the next block, otherwise, if $\bigcup_{p,q \in A} p(p,q) \neq V(C)$ add a vertex to $S$ the same way as described for odd cycles at the previous paragraph.

After processing all blocks, if $|S| \geq 3$, then $S$ is a minimum strong geodetic set of $G$ and the algorithm finishes. Otherwise, verify whether $G$ contains any block that is an even cycle, if so, add an arbitrary vertex of $G$ to $S$ and finish. Otherwise, return~$S$ and finish.

\begin{theorem}\label{thm:cacti}
The algorithm presented above is correct.
\end{theorem}
\begin{proof}
For now consider that the algorithm receives as input a cactus graph $G = (V, E)$ whose cut-tree $T = (V', E')$ contains at least 3 leaves. We will first show that the returned set $S$ is a strong geodetic set. From the description of the algorithm we know that we will have at least one vertex in $S$ for each leaf of $T$. Let $F$ be the set of leaves of $ T $ and $ f_1 \in F $ a leaf that represents an edge $e = ux$ in $G$ whose simplicial vertex is $u$. And let $f_2$ be another leaf of $T$, with $v \in S \cap V(F_2)$, finally note that any path between $u$ and $v$ contains $x$, covering all vertices of $e$.

Now let $f_1$ be a leaf of $T$ that represents an even cycle $C$ of length $l$. By the algorithm, we add to $S$ a vertex $v$ whose distance to the cycle's cut-vertex $a$ is $\dfrac{l}{2}$, thus, we have two distinct paths between $ v $ and $ a $ with length $\dfrac{l}{2}$: $c_1$ and $c_2$. Let $f_2$ and $f_3$ be two other leaves of $T$, that exist by hypothesis. Any shortest path that goes from~$v$ to the cited leaves contains $a$, so we set the shortest path between $f_1$ and $f_2$ to pass through $c_1$ and the shortest path between $f_1$ and $f_3$ to pass through~$c_2$, covering all vertices of $C$.

Now let $f_1$ be a leaf of $T$ representing an odd cycle $ C $ of length $ l $. By the algorithm, we add two vertices to $ S $: $ v_1 $ and $ v_2 $, such that their distances to the cycle's cut-vertex~$ a $ are the same: $\left\lfloor \dfrac{l}{2} \right\rfloor$. Observe that: $p(v_1,a) \cup p(v_2,a) =  V(C)$, thus, by choosing any shortest path from $ v_1 $ to another vertex $v_3 \in S \cap f_3$, where $f_3$ is another leaf of $T$, and from $ v_2 $ to the same leaf $ f_3 $ all vertices of $C$ will be covered.

Let $t \in T$ be an internal vertex of $ T $ that represents an edge $e = uv$ of $G$. Let $ C_1 $ and $ C_2 $ be connected components of $T - \{t\}$. In addition, consider $ f_1 $ to be a leaf of $ C_1 $ and $ f_2 $ a leaf of $ C_2 $, now note that any path between $x \in S \cap V(f_1)$ and $y \in S \cap V(f_2)$ contains $ u $ and $ v $. Therefore, all vertices of $e$ will be covered.

Let $t \in T$ be an internal vertex of $ T $ which represents a  cycle $ C $ of size $ l $ at $G$. Let $ A $ denote the set of cut-vertices of $ C $, the algorithm verifies whether $\bigcup_{p,q \in A} p(p,q) = V(C)$, we claim that if that holds, then all vertices of $ C $ are covered by shortest paths between vertices in $ S $. In fact, let $ v $ be any vertex of $ C $, assuming $\bigcup_{p,q \in A} p(p,q) = V(C)$, there are vertices $a_1, a_2 \in A$ such that~$v \in p(a_1,a_2)$. Now, let $ C_1 $ and $ C_2 $ be the two connected components of $G - \{a_1\}$ such that $ C_1 $ is the one that has no vertex of $C$. Analogously, let $ C_3 $ and $ C_4 $ be the two connected components of $G - \{a_2\}$ such that $ C_3 $ is the one that has no vertex in $C$. Let $x \in C_1 \cap S$ and $y \in C_3 \cap S$, these vertices exist because the algorithm guarantees that every leaf of $T$ has a vertex in $S$, observe that any shortest path between $ x $ and $ y $ contains $ v $. Nevertheless, if $\bigcup_{p,q \in A} p(p,q) \neq V(C)$ and there are no vertices $i, j \in A$ such that $D(i,j) = \dfrac{l}{2}$, then the algorithm adds a vertex $ v \in V(C) $ to $ S $ so that there exists vertices $ a_1, a_2 \in A $ such that $D(a_1,v) - D(a_2,v) \leq 1$. Thus, it holds that $p(a_1,a_2) \cup p(a_1,v) \cup p(a_2,v) = V(C)$, having all vertices of $C$ covered. Finally, if there are vertices~$i, j \in A$ such that $D(i,j) = \dfrac{l}{2}$, then it is possible to cover all vertices of $C$, since $T$ has at least 3 leaves.

Now, it remains to argue that the returned set $S$ is minimum. Observe that odd cycles situated at leafs of $T$ must have at least 2 of its vertices in $S$ and even cycles situated at leafs of $T$ must have at least 1 of its vertices in $S$. Edges located at leaves of $T$ must have its simplicial vertex added to $S$. Now, observe that for internal vertices of $T$ we add to $S$ the minimum amount of vertices needed, that is, for edges we add none, for cycles that can be covered by shortest paths between its cut-vertices we add none, and for cycles that cannot be covered that way we add a vertex to $S$, which is the minimum required.
\end{proof}

\begin{corollary}\label{thm:rec_cacto_P}
There is a polynomial-time algorithm that solves the \textsc{strong geodetic recognition} problem for cacti graphs.
\end{corollary}
\begin{proof}
In order to verify that a given vertex set~$X$ of a cactus graph~$G = (V, E)$ is a strong geodetic set, we utilize the reduction presented in Proposition~\ref{prop:SGRtoSG}. If the reduction is applied to $G$, then a cactus graph $G'= (V',E')$ arises, this occurs because the reduction only adds one-degree vertices to the graph. Thus, it is possible to solve the SGR for cacti graphs by solving the SG at a related cactus. Finally, since it is possible to solve the~SG for cacti graphs in polynomial time, then~SGR for cacti graphs is also computable in polynomial time.
\end{proof}

    Now we present polynomial-time algorithms to SGR restricted to graphs of diameter~2 and restricted to \textit{split graphs}, which are those whose vertex set can be partitioned into a clique and an independent set. Both proofs follow a similar approach.

    \begin{theorem}\label{teorema_rec_diameter}
        Let~$G = (V,E)$ be a connected graph of diameter~2 and consider~$S \subseteq V$.
        There exists an~$\mathcal{O}(|S|^2 \, \cdot \, |V \setminus S|)$-time algorithm that decides whether~$S$ is a strong geodetic set of $G$.
    \end{theorem}
    \begin{proof}
        At first, we construct an auxiliary bipartite graph~$H=(A,B,E')$, with parts~$A = \{v_{i,j} \mid i,j \in S \wedge i \neq j \}$ and~$B = V \setminus S$. In addition, there is an edge between~$v_{i,j} \in A$ and~$y \in B$ if and only if~$(i,y,j)$ is an $i,j$-shortest path in~$G$.
        
        Now, we compute a maximum matching~$M$ of~$H$.
        This can be done in time~$\mathcal{O}\left(|E'|\right)$~\cite{alom2010finding}.
        Observe that~$|A| \leq |S|^2$ and~$|B| = |V \setminus S|$, then it is possible to compute such a matching in time $\mathcal{O}\left(|S|^2 \cdot |V \setminus S|\right)$.
        Finally, if~$|M| = |B|$, then output YES, otherwise, output NO.
        
        In order to prove the correctness of the algorithm we prove that~$M$ has size~$|B|$ if and only if~$S$ is a strong geodetic set of~$G$.
        Assume that~$|M| = |B|$, then for each vertex~$b \in B$ there is an edge~$v_{i,j}b \in M$ and we use the~$(i,b,j)$ shortest path to cover $b$.
        Moreover, since~$M$ is a matching, for each pair of vertices~$i,j \in S$ it will be assigned a unique $i,j$-shortest path in~$I(S)$.
        Finally, if there are still shortest paths to be assigned in~$I(S)$, any choice of shortest paths will guarantee a valid strong geodetic set~$S$.
        
        For the converse, assume that~$S$ is a strong geodetic set of~$G$, then there is a shortest path choice~$I(S)$ that covers all vertices in~$V \setminus S$.
        Let~$u \in V \setminus S$ and let~$M$ be an empty set.
        It holds that at least one~$p,q$-shortest path in~$I(S)$ covers~$u$, we add the edge~$v_{p,q}u$ to~$M$, observe that~$v_{p,q}u \in E'$, by the definition of~$H$.
        Repeat this process for every~$u \in V \setminus S$.
        It results that~$M$ is a maximum matching of~$H$, with~$|M| = |B|$.
        In fact, note that~$M$ has exactly one edge incident to each vertex in~$B$ and at most one edge in~$M$ is incident to a vertex in~$A$, given that there is a unique shortest path in~$I(S)$ for each vertex pair of~$S$.
    \end{proof}

    Observe that this result illustrates an interesting discrepancy between SG and~SGR: SG restricted to 2-diameter graphs is \NP-complete and~SGR restricted to 2-diameter graphs can be solved in polynomial-time.

    \begin{theorem}\label{teorema_rec_split}
        Let $G = (V,E)$ be a connected split graph and consider~$S \subseteq V$.
        There exists an~$\mathcal{O}\left(|S|^2 \cdot |V \setminus S|\right)$-time algorithm that decides whether~$S$ is a strong geodetic set of~$G$ (SGR).
    \end{theorem}
    \begin{proof}
        We propose a construction that follows the same approach of Theorem~\ref{teorema_rec_diameter}.
        Create an auxiliary bipartite graph~$H=(A,B,E')$, with~$B = V \setminus S$.
        Now it remains to define~$A$ and~$E'$: for each vertex pair~$(i,j)$, with~$i,j \in S$ and~$i \neq j$~do:
        \begin{itemize}
            \item If $D(i,j) \neq 3$, add a vertex $v_{i,j}$ to $A$. In addition, add the edges $v_{i,j}k$ for all~$k \in B$ such that $(i,k,j)$ is a shortest path in $G$.
            \item If $D(i,j) = 3$, add the vertices $v_{i,j}$ and $\overline{v_{i,j}}$ to $A$. Then, add the edges $v_{i,j}k$ for all $k \in N(i) \cap B$, and add the edges $\overline{v_{i,j}}k'$ for all $k' \in N(j) \cap B$.
        \end{itemize}
        Now, we compute a maximum matching $M$ of $H$ in time $\mathcal{O}(|S|^2 \cdot |V \setminus S|)$, the time complexity is derived similarly as in Theorem~\ref{teorema_rec_diameter}.
        Finally, if $|M| = |B|$ output~YES, otherwise, output NO.
        
        In order to prove the correctness of the algorithm we prove that the maximum matching $M$ of $H$ has size $|B|$ if and only if $S$ is a strong geodetic set of $G$.
        Assume that $|M| = |B|$, then, for each vertex $b \in B$ there is an edge $ab \in M$, with $a \in A$. If $a = v_{i,j}$, with $D(i,j) = 2$, then assign the $(i,b,j)$ shortest path to $I(S)$.
        On the other hand, if $a = v_{i,j}$ (without loss of generality), and $D(i,j) = 3$, we set $b$ to be on the $i,j$-shortest path, and the other vertex present on the $i,j$-shortest path will be the vertex in $B$ that is an endpoint of the edge matching $\overline{v_{i,j}}$.
        Finally, since $M$ is a matching, for each pair of vertices $i,j \in S$ it will be assigned a unique $i,j$-shortest path in $I(S)$.
        Therefore, $I(S)$ defines a strong geodetic set $S$.
        
        For the converse, assume that $S$ is a strong geodetic set for $G$ defined by~$I(S)$.
        Let $M$ be an empty set.
        Then, for each $i,j$-shortest path $(i,k,j)$ of size 2 in $I(S)$, with $k \in B$, add the edge~$v_{i,j}k$ in~$M$.
        And for each $i,j$-shortest path $(i,k,l,j)$ of size 3 in $I(S)$, add the edges $v_{i,j}k$, if $k \in B$, and~$\overline{v_{i,j}}l$, if $l \in B$.
        Now, remove edges of $M$ until there is exactly one edge in $M$ incident to each vertex in~$B$.
        Finally, observe that $M$ is a maximum matching of $H$, with $|M| = |B|$.
    \end{proof}

\section{Further Research}
\label{sec:conclusions}

    Our results show that the complexity of the decision version of the {\sc Geodetic} problem and SGR are quite similar.
    Both are \NP-complete for co-bipartite and bipartite graphs (remember that {\sc Geodetic Number} is \NP-hard for chordal bipartite graphs~\cite{dourado2010some}), while they are tractable on split graphs and we strongly believe the same on cacti graphs.
    Moreover, both are intractable for graphs of bounded maximum degree.
    The first problem we leave is about the complexity of~SGR on subcubic graphs, while its intractability is known for {\sc Geodetic Number}, and for bipartite graphs of maximum degree~4  on SGR, as we have proved.
    However, the complexities differ for graphs of diameter~2, being \NP-hard on {\sc Geodetic Number} and polynomial on~SGR.
    So, what is the complexity of~SGR for graphs of diameter~3?
    On the opposite way, as {\sc Geodetic Number} is tractable on cographs, would be interesting to prove the same for~SG, or even~SGR.
    Another question left is the complexity of~SG for split graphs, since it is tractable for the other two problems.
    Moreover, with respect to parameterized complexity, is~SGR in \FPT~when parameterized only by the size of the given vertex set?
    

\bibliographystyle{abbrv}
\bibliography{ref}


\end{document}